\documentclass[runningheads]{llncs}

\usepackage[utf8]{inputenc} 
\usepackage[T1]{fontenc} 

\usepackage{amsmath}
\usepackage{amssymb}
\usepackage{mathtools}
\usepackage{amsfonts} 

\usepackage{mdframed}
\usepackage[most]{tcolorbox}
\usepackage{nicefrac} 
\usepackage{microtype} 
\usepackage{lipsum}

\usepackage{hyperref}
\usepackage{url} 
\usepackage{xspace}

\usepackage{graphicx}
\usepackage{xcolor}
\graphicspath{{media/}}
\usepackage{booktabs} 
\usepackage{multirow}
\usepackage[ruled,linesnumbered]{algorithm2e}

\usepackage{todonotes}
\usepackage{fancyhdr} 
\usepackage{lineno}

\newcommand{\xp}{{\sf XP}\xspace}
\newcommand{\fpt}{{\sf FPT}\xspace}
\newcommand{\npp}{\textup{\textsf{P}}\xspace}
\newcommand{\np}{\textup{\textsf{NP}}\xspace}

\newcommand{\ptas}{{\sf PTAS}\xspace}
\newcommand{\nph}{{\textsf{NP}\textrm{-hard}\xspace}}
\newcommand{\npc}{{\textsf{NP}\textrm{-complete}\xspace}}
\newcommand{\npcs}{{\textsf{NP}\textrm{-completeness}\xspace}}

\newcommand{\wth}{\textsf{W[2]}\textrm{-hard}\xspace}

\newcommand{\wtot}{\textsf{W[1]}\xspace}

\newcommand{\lapxh}{\textrm{log-APX-hard}\xspace}
\newcommand{\lapxxh}{\textrm{log-APX-hardness}\xspace}

\newcommand{\apxh}{\textrm{APX-hard}\xspace}

\newcommand{\mcs}{\textup{\textsc{MCS}}\xspace}
\newcommand{\mscs}{\textup{\textsc{MSCS}}\xspace}
\newcommand{\mcss}{\textup{\textsc{MCSS}}\xspace}
\newcommand{\mss}{\textup{\textsc{MSS}}\xspace}

\newcommand{\udg}{\textup{\textsc{UDG}}\xspace}
\newcommand{\mds}{\textup{\textsc{MDS}}\xspace}
\newcommand{\maxsat}{\textup{\textsc{MAX-3SAT(8)}}\xspace}

\newcommand{\NN}{\mathrm{\hat{N}}\xspace}
\newcommand{\n}{\mathrm{N}\xspace}
\newcommand{\dist}{\mathrm{d}\xspace}

\begin{document}
\title{Minimum Selective Subset on Unit Disk Graphs and Circle Graphs}
%
%
\author{%
Bubai Manna
}%
\institute{
Indian Institute of Technology Kharagpur, India\\
\email{bubaimanna11@gmail.com}}

\authorrunning{B. Manna}
\maketitle              
\begin{abstract}
In a connected simple graph \( G = (V(G),E(G)) \), each vertex is assigned one of \( c \) colors, where \( V(G) = \bigcup_{\ell=1}^{c} V_{\ell} \) and \( V_{\ell} \) denotes the set of vertices of color \(\ell\). A subset \( S \subseteq V(G) \) is called a \emph{selective subset} if, for every \( \ell \), \( 1 \leq \ell \leq c \), every vertex \( v \in V_{\ell} \) has at least one nearest neighbor in 
\( S \cup (V(G) \setminus V_{\ell}) \) that also lies in \( V_{\ell} \). 
The \emph{Minimum Selective Subset} (\mss) problem asks for a selective subset of minimum size.

We show that the \mss problem is \lapxh on general graphs, even when \( c = 2 \). As a consequence, the problem does not admit a polynomial-time approximation scheme (\ptas) unless \(\npp = \np\). On the positive side, we present a \ptas for unit disk graphs that does not require a geometric representation and applies for arbitrary $c$. We further prove that \mss remains \npc\ in unit disk graphs for arbitrary \(c\). In addition, we show that the \mss problem is \apxh on circle graphs, even when $c=2$ \footnote{This work appeared in CALDAM 2026 \cite{caldam2026}}.

\end{abstract}
\section{Introduction}
Many computational tools have been developed for supervised learning methods on a labeled training set \( T \) embedded in a metric space \( (X,d) \). Each data point \( t \in T \) is associated with a label (also called a \emph{character} or \emph{color}), chosen from a set \( C = \{1,2,\dots,c\} \). The objective is to extract a smallest possible subset \( S \subseteq T \) such that every point in \( T \) either belongs to \( S \) or has at least one nearest neighbor (with respect to the metric $d$) within \( S \) that shares the same character. This optimization problem, called the \emph{Minimum Consistent Subset} (\mcs), was originally formulated by Hart~\textcolor{blue}{\cite{Hart}} in 1968, which has received thousands of citations, highlighting its significant impact in the field. However, the paper~\textcolor{blue}{\cite{Hart}} did not establish any complexity results or algorithms.

Later in 1991, Wilfong~\textcolor{blue}{\cite{Wilfong}} defined two problems \mcs and \mss together and proved that the \mcs and \mss problems are $\npc$ in \(\mathbb{R}^2\) for $c\geq 3$ and $c\geq 2$, respectively. It also proposed a polynomial-time algorithm when there is only one red point and all other points are blue in $\mathbb{R}^2$. Later, in 2018, it was proved that \mcs remains $\npc$ when $c=2$ in \(\mathbb{R}^2\)~\textcolor{blue}{\cite{Khodamoradi}}. Recently, Banerjee et al.~\textcolor{blue}{\cite{BBC}} showed that \mcs is \wth (for arbitrary $c$) and \mss is \wtot (for $c=2$), both parameterized by the solution size. Various algorithms, including those for many restricted inputs for the \mcs problem in \(\mathbb{R}^2\) have been proposed~\textcolor{blue}{\cite{BBC,Ahmad,Chitnis22}}, highlighting its significance in machine learning and computational geometry. The only algorithm for the \mss problem is a \ptas, which was established when $c=2$~\textcolor{blue}{\cite{BBC}}.

The Minimum Selective Subset (\mss) problem plays a crucial role in optimizing data selection by identifying the smallest subset which preserves essential information. It can be viewed both as a clustering and a proximity problem. This is particularly useful in applications such as fingerprint recognition, character recognition, and pattern recognition, where it helps reduce redundancy and improve decision-making in classification and feature selection tasks. \mss was introduced because the standard \mcs methods, like Hart's~\textcolor{blue}{\cite{Hart}}, could not guarantee that the resulting subset was the smallest possible size. \mss applies a stronger consistency condition on the subset to create a mathematical framework that is amenable to an exact minimization procedure, thereby fulfilling the optimization goal that \mcs heuristics failed to satisfy. So far, we have discussed \mcs and \mss problems along with their published results in $\mathbb{R}^2$. We now turn to these problems in the context of graph algorithms.

Banerjee et al.~\textcolor{blue}{\cite{BBC}} proved that \mcs is \wth~\textcolor{blue}{\cite{book}} when parameterized by the solution size, even with only two colors on general graphs. Dey et al.~\cite{DeyMN21,DeyMN23} provided polynomial-time algorithms for \mcs on some simple graph classes including paths, spiders, caterpillars, combs, and trees (for trees, $c=2$). \xp, $\npc$, and \fpt (when $c$ is a parameter) results on trees, can be found in~\textcolor{blue}{\cite{Arimura23,aritra}}. The \mcs problem is also $\npc$ on interval graphs~\textcolor{blue}{\cite{aritra}} and $\apxh$ on circle graphs~\textcolor{blue}{\cite{2024arXiv240514493M}}. Variants, such as the \textit{Minimum Consistent Spanning Subset} (\mcss) and the \textit{Minimum Strict Consistent Subset} (\mscs) of \mcs, have been studied on trees~\textcolor{blue}{\cite{banik2024minimum,biniaz2024minimum,manna2024minimumstrictconsistentsubset}}. However, the algorithmic results for \mss have not been extensively studied to date. Banerjee et al.~\textcolor{blue}{\cite{BBC}} only showed that \mss is $\npc$ on general graphs. Very recently, the \mss problem has been studied in various settings, including $\mathcal{O}(\log n)$-approximation algorithms for general graphs, $\npc$ results for planar graphs, and linear-time algorithms for trees and unit interval graphs~\textcolor{blue}{\cite{manna2025minimumselectivesubsetgraph}} (published in CCCG 2025).

\textbf{Our Contributions.} The \mss problem admits an $\mathcal{O}(\log n)$-approximation on general graphs, which raises the question of whether better approximations exist. We show in Section~\textcolor{blue}{\ref{lapxhard}} that \mss is \lapxh even when $c=2$. Hence, the problem is also \apxh and does not admit a \ptas on general graphs. This leads to the natural question of whether some graph classes allow a \ptas. To date, none are known. We answer this by proving in Section~\textcolor{blue}{\ref{ptas}} that \mss admits a \ptas on unit disk graphs for arbitrary $c$, without requiring a geometric representation. Unit disk graphs are also fundamental in wireless networks, robotics, and computational geometry, where efficient approximation algorithms are highly relevant~\textcolor{blue}{\cite{GOLDIN2009234}}. Before presenting our \ptas, we establish in Section~\textcolor{blue}{\ref{nphard}} that \mss remains $\npc$ on unit disk graphs when $c$ is arbitrary. We also investigate whether \mss is \apxh in other graph classes. In Section~\textcolor{blue}{\ref{apxhard}}, we prove that \mss is \apxh on circle graphs even when $c=2$. Circle graphs, which model intersecting chords, have applications in VLSI design, scheduling, and bioinformatics~\textcolor{blue}{\cite{article,sherwani2012algorithms}}. All proofs of results marked with $(^{\ast})$ can be found in the Appendix. 

\section{Preliminaries}\label{preli}
Let $G=(V(G),E(G))$ be a graph, where $V(G)$ is the vertex set and $E(G)$ is the edge set. For any $U\subseteq V(G)$, \( G[U] \) denotes the subgraph of \( G \) induced on \( U \), and \( \lvert U \rvert \) is the cardinality of $U$. We denote \( [n] \) as the set of integers \( \{1,\ldots, n\} \). We use an arbitrary vertex color function \( C:V(G)\rightarrow [c] \), which assigns each vertex exactly one color from the set \( [c] \). For a subset of vertices \( U \subseteq V(G)\), let \( C(U) \) represent the set of colors of the vertices in \( U \), formally defined as \( C(U) = \{C(u) \mid u \in U\} \). The shortest path distance (i.e., \textit{hop-distance}) between two vertices \( u \) and \( v \) in \( G \) is denoted by \( \dist (u,v) \). Distance between \( v\in V(G) \) and the set \( U\subseteq V(G) \) is given by  $\dist(v,U) = \min_{u\in U} \dist (v,u)$. Similarly, the distance between two subgraphs \( G_1 \) and \( G_2 \) in \( G \) is defined as $\dist (G_1,G_2) = \min \{\dist (v_1,v_2) \mid v_1 \in V(G_1), v_2 \in V(G_2)\}$. The set of nearest neighbors of \( v \) in the set \( U \) is denoted as \( \NN(v,U) \), formally defined as $\NN(v,U) = \{u \in U \mid \dist (v,u) = \dist (v,U) \}$. Therefore, if $v\in U$, then $\NN (v, U)=\{v\}$. The set of vertices in $U$ adjacent to \( v \) is given by  
\(
\n (v,U) = \{ u \in U \mid (u,v) \in E(G) \}.
\)  
We also define  
\(
\n[v,U] = \{ v \} \cup \n(v,U).
\) For any two subsets \( U_1, U_2 \subseteq V(G) \), we define  
\(
\n(U_1,U_2) = \bigcup_{v\in U_1} \n(v, U_2),
\)  
and  
\(
\n[U_1,U_2] = \bigcup_{v\in U_1} \n[v, U_2].
\)  
Most symbols and notations follow standard conventions from~\textcolor{blue}{\cite{diestel2012graph}}. Suppose \( G = (V(G), E(G)) \) is a given simple connected undirected graph where $\bigcup_{i=1}^{c} V_i = V(G)$ and $V_i \cap V_j = \emptyset$ for $i\neq j$ and each vertex in \( V_i \) is assigned color \( i \). A \emph{Minimum Consistent Subset} (\mcs) is a subset \( S \subseteq V(G) \) of minimum cardinality such that for every vertex \( v \in V(G) \), if \( v \in V_i \), then $\NN(v,S) \cap V_i$ is non-empty.  
 \begin{definition}[Selective Subset]
A subset \( S \subseteq V(G) \) is called a \emph{Selective Subset (\mss)} if, for each vertex \( v \in V(G) \), if \( v \in V_i \), the set of nearest neighbors of \( v \) in \( S \cup (V(G) \setminus V_i) \), denoted as $\NN(v, S \cup (V(G) \setminus V_i))$, contains at least one vertex \( u \) such that \( C(v) = C(u) \). An \mss is a selective subset of minimum cardinality. The decision version of the \mss problem is as follows:
\begin{tcolorbox}[
    enhanced,
    title={\color{black} \sc{Decision Version of Selective Subset Problem on Graphs}},
    colback=white,
    boxrule=0.4pt,
    attach boxed title to top center={xshift=0cm, yshift*=-2mm},
    boxed title style={size=small, frame hidden, colback=white}]
    \textbf{Input:} A graph \( G=(V(G),E(G)) \), a coloring function \( C:V(G)\rightarrow [c] \), and an integer \( s \).\\
    \textbf{Question:} Does there exist a selective subset of size at most \( s \) for \( (G,C) \)?
\end{tcolorbox}
In other words, we seek a vertex set \( S \subseteq V(G) \) of minimum cardinality such that every vertex \( v \) has at least one nearest neighbor of the same color in the graph, excluding vertices of the same color as \( v \) that are not in \( S \). If all the vertices of a graph \( G \) are of the same color (i.e., \( G \) is monochromatic), then any vertex in the graph forms a valid \mss. Figure~\textcolor{blue}{\ref{bubai1}} illustrates an example of \mss. 
\end{definition}
\begin{figure}[ht]
\includegraphics[width=8.5cm]{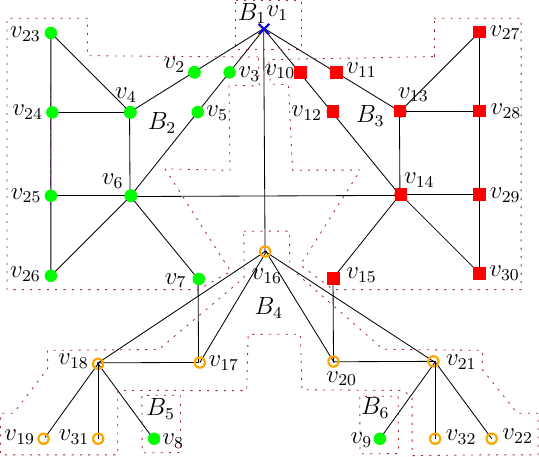}
\centering
\caption{Colors: \emph{blue $\equiv$ cross}, \emph{green $\equiv$ disk}, \emph{red $\equiv$ square}, and \emph{orange $\equiv$ circle}. $V(G)=$ $V_{blue}$ $\cup $ $ V_{green}$ $\cup$ $ V_{red}\cup V_{orange}$, where $V_{blue}=$ $\{v_1\}$, $V_{green}$ $=$ $\{v_2$ $,\dots,$ $v_{9},v_{23},$ $\dots,v_{26}\}$, $V_{red}=\{v_{10},\dots,v_{15},v_{27},\dots,v_{30}\}$, and $V_{orange}=\{v_{16},\dots,v_{22},v_{31},v_{32}\}$. $S=$ $\{v_1,$ $v_2,$ $v_3,$ $v_7,$ $v_8,$ $v_{9},v_{10},v_{11},v_{15},v_{16}\}$ is an \mss, and $S=\{v_1,$$v_4,$ $v_5,$ $v_7,v_8,v_{9},v_{12},v_{13},v_{15},v_{16}\}$ is also an \mss. Brown-dotted regions indicate the blocks. The complete list of blocks is $B_1=\{v_1\}$, $B_2=\{v_2,\dots,v_7,v_{23},\dots,v_{26}\}$, $B_3=\{v_{10},\dots,v_{15},v_{27},\dots,v_{30}\}$, $B_4=\{v_{16},\dots,v_{22},v_{31},v_{32}\}$, $B_{5}=\{v_{8}\}$, $B_{6}=\{v_{9}\}$. $B_{2,1}=\{v_2,v_3,v_6,v_7\}$, $B_{2,2}=\{v_4,v_5,v_{25},v_{26}\}$. $\{\{v_2\}, \{v_3\},\{v_7\}\}$ is a collection of 2-distance sets in $B_{2,1}$.}\label{bubai1}
\end{figure}
\begin{definition}[Block]
A block is a maximal connected subgraph whose vertices all have the same color (i.e., a maximal connected monochromatic subgraph).
\end{definition}
Figure~\textcolor{blue}{\ref{bubai1}} illustrates an example of the blocks.
Suppose $B_1,\dots,B_k$ is the complete list of blocks in $G$. We assume that $\lvert V(G)\rvert =n$, so that $k\leq n$. We form the sets $B_{i}^1$, $B_{i}^2$, $B_{i,3}$ for each $i=1,\dots,k$ as follows (see Figure~\textcolor{blue}{\ref{bubai1}}):
\begin{itemize}
    \item Initially, $B_{i,1}:=\emptyset $, $B_{i,2}:=\emptyset$.
    \item For each vertex $v\in B_i$, if there exists a vertex $u\in \n (v,V(G))$ such that $C(u)\neq C(v)$, then $v\in B_{i,1}$.
    \item For any vertex $v\in B_i\setminus B_{i,1}$ if $\dist (v,B_{i,1})=1$, then $v\in B_{i,2}$.
    \item We denote $B_{i,3}=B_{i,1}\cup B_{i,2}$.
\end{itemize}
\begin{lemma}\label{lemma2}$^{\ast}$
For any vertex $v\in B_{i,1}$ and a selective subset $S$, we have $\n [v,B_{i,3}]\cap S\neq \emptyset$ for $1\leq i\leq k$.
\end{lemma}

\section{\lapxxh of \mss on General Graphs}\label{lapxhard}
We establish a reduction from the \textsc{Minimum Dominating Set} 
(\mds) problem to the \mss problem. In the \textsc{Minimum Dominating Set} problem, the input is a graph \( G \) together with an integer \( s \). The task is to decide whether there exists a subset \( D \subseteq V(G) \) of size at most \( s \) such that every vertex \( u \in V(G) \) satisfies \( \n [u, V(G)] \cap D \neq \emptyset \). It is well known that the \textsc{Minimum Set Cover} problem is \lapxh (see~\textcolor{blue}{\cite{10.5555/1965254}} for complexity class definitions), and moreover, it is $\nph$ to approximate it within a factor of \( \delta \cdot \log n \) for some positive constant \( \delta \)~\textcolor{blue}{\cite{raz97}}. Since there exists an \( L \)-reduction from the \textsc{Minimum Set Cover} problem to the \textsc{Minimum Dominating Set} problem, the \textsc{Minimum Dominating Set} problem is also \lapxh.

Let \( (G, s) \) be an arbitrary instance of the \textsc{Minimum Dominating Set} problem. We construct an instance \( (G', C, s+1) \) for the \mss problem as follows (see Figure~\textcolor{blue}{\ref{bubaix}(a)}). Define the new graph \( G' \) with \( V(G') = V(G) \cup \{z\} \) and \( E(G') = E(G) \cup \{(z, u) \mid u \in V(G)\} \). The color function \( C \) assigns color 1 to all vertices \( u \in V(G) \), and color 2 to the additional vertex \( z \).
\begin{lemma}\label{bubailemma6}$^{\ast}$
$G$ has a dominating set of size at most $s$ if and only if $G'$ has a selective subset of size at most $s+1$.
\end{lemma}
\begin{figure}[ht]
\includegraphics[width=8.4cm]{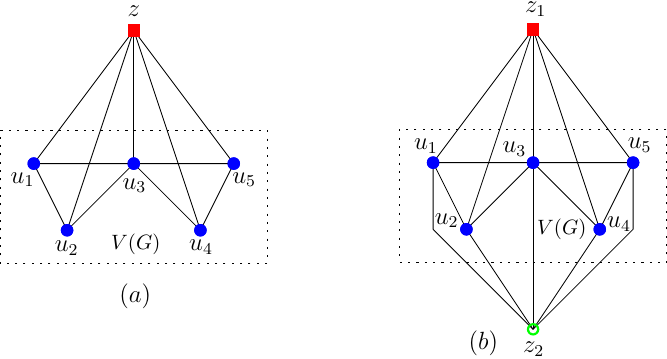}
\centering
\caption{Colors: \emph{blue $\equiv$ disk}, \emph{red $\equiv$ square}, and \emph{green $\equiv$ circle}. (a) Reduction from an instance of \textsc{Minimum Dominating Set} problem to an instance of \mss problem when $c=2$. (b) Example of the reduction when $c=3$.}\label{bubaix}
\end{figure}
\begin{theorem}\label{th1}$^{\ast}$
There exists a constant $\delta >0$ such that it is $\nph$ to approximate the \mss problem within a factor of $\delta \cdot \log n$, where $n$ is the number of vertices in the graph.
\end{theorem}
\begin{remark}
Note that the above reduction remains valid even if we add any number of new vertices (each adjacent to all of $V(G)$) and assign each a distinct color. The correctness of Lemma~\textcolor{blue}{\ref{bubailemma6}} and the resulting hardness theorem continue to hold under this extended construction (see Figure~\textcolor{blue}{\ref{bubaix}(b)}).
\end{remark} 
\section{$\npcs$ of \mss on Unit Disk Graphs}\label{nphard}
A graph \( U = (V(U), E(U)) \) is called a \textit{unit disk graph} (\udg) if its vertices can be represented as points in the Euclidean plane such that an edge exists between two vertices if and only if their Euclidean distance is at most $2$. 

Formally, \( U \) is a UDG if there exists a mapping \( f: V \to \mathbb{R}^2 \) such that:  
\begin{equation}
    (u, v) \in E(U) \iff \| f(u) - f(v) \| \leq 2
\end{equation}
where \( \| \cdot \| \) denotes the Euclidean norm. 

Clark et al.~\textcolor{blue}{\cite{Clark1990}} showed that \textsc{Minimum Dominating Set} (\mds) is $\npc$ in \udg. We reduce from an instance $U$ of the \udg with $\lvert V(U) \rvert = n$ to an instance $U'$ as follows: 
\begin{figure}[ht]
\includegraphics[width=8cm]{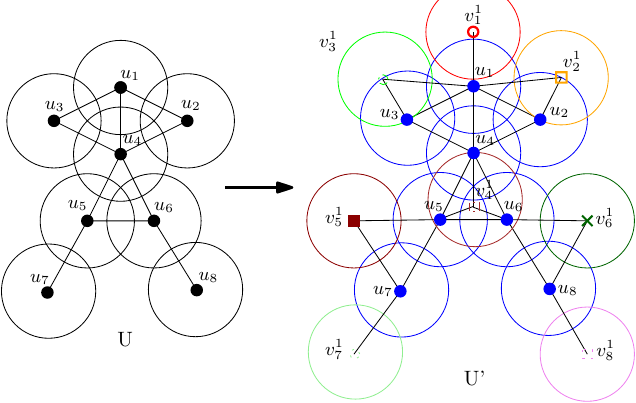}
\centering
\caption{Colors: \emph{blue $\equiv$ disk}, \emph{red $\equiv$ circle}, \emph{orange $\equiv$ fsquare}, \emph{green $\equiv$ dash dotted circle}, \emph{brown $\equiv$ dash dotted fsquare}, \emph{darkred $\equiv$ square}, \emph{darkgreen $\equiv$ cross}, \emph{lightgreen $\equiv$ dotted circle}, and \emph{violet $\equiv$ dotted square}. An example of the reduction when $m=1$. Each \textit{blue} disk in \( U' \) is adjacent to a disk of a distinct color, different from all other colors in \( U' \).}\label{bubai2} 
\end{figure}

\textbf{Reduction.} Define $V(U') = V(U) \cup V(X)$ and $E(U') = E(U) \cup E(X)$, where we introduce a subgraph $X$ with vertex set $V(X)$ and edge set $E(X)$ as follows (see Figure~\textcolor{blue}{\ref{bubai2}}). Initially, set $V(X) := \emptyset$ and $E(X) := \emptyset$. Assign color $0$ to all unit disks in $V(U)$, i.e., $C(V(U)) = \{0\}$. For each $u_i \in V(U)$, introduce a total of $m$ (where $m\in \mathbb{N}$) unit disks $v_i^{1},\dots,v_{i}^{m} \in V(X)$. Additionally, place $v_i^{1},\dots,v_i^{m}$ in such a way that $u_i$ and $v_i^{l}$ are adjacent for $1\leq l\leq m$ (where $v_i^{l}$ may also be adjacent to other unit disks). These new edges are added to $E(X)$. All unit disks in $V(X)$ have distinct colors, none of which is $0$. Thus, we obtain $\lvert V(X)\rvert =mn$, $\lvert V(U') \rvert = nm+n$, and $\lvert C(V(U')) \rvert = nm+1$.
\begin{lemma}\label{npclemma}$^{\ast}$
$U$ has a dominating set of size $t$ if and only if $U'$ has a selective subset of size $nm+t$.
\end{lemma}
\begin{theorem}\label{th2}$^{\ast}$
\mss is $\npc$ on unit disk graphs.
\end{theorem}

\section{\ptas of \mss on Unit Disk Graphs}\label{ptas}
A unit disk graph may admit multiple geometric representations. In this work, we assume that the geometric representation \( f \) is either unknown or not explicitly given. We first describe our approach for a general graph before restricting to unit disk graphs. If \( G \) is not a connected graph, we apply the algorithm independently on each component; hence, we may assume that \( G \) is connected. 

The key idea is that blocks are independent within any selective subset solution.
Exploiting this property, we compute \emph{local} selective subsets within each block independently. For the set $B_{i,1}$ corresponding to a given block $B_i$, we define a family of sets $D_{i,1}^{1},\dots,D_{i,1}^{t_i}$ such that $\dist(D_{i,1}^{j},D_{i,1}^{\ell})>2$ for all $j \neq \ell$. By Lemmas~\textcolor{blue}{\ref{lemma2}} and~\textcolor{blue}{\ref{lemma3}}, this ensures that no two sets $D_{i,1}^{j}$ and $D_{i,1}^{\ell}$ share a common vertex in optimal solution. Consequently, the solutions for $D_{i,1}^{1},\dots,D_{i,1}^{t_i}$ provide a lower bound on the size of the optimal solution. However, the union of these solutions does not necessarily form a valid solution for the block $B_{i}$. To address this, we enlarge each set $D_{i,1}^{j}$ into a corresponding set $E_{i,1}^{j}$, ensuring that the union of the solutions for $E_{i,1}^{1},\dots,E_{i,1}^{t_i}$ yields a valid solution for $B_i$. We refer to the solutions for $D_{i,1}^{j}$ and $E_{i,1}^{j}$ as \emph{local solutions}. By combining these local solutions, we obtain a blockwise selective subset, and taking the union over all blocks 
yields a global solution used in our \ptas.  

For any subset $U \subseteq V(G)$, we denote its minimum selective subset (\emph{local}) by $S^{\min}(U)$ and a selective subset (\emph{local}) of $U$ by $S(U)$. The global optimum is denoted by $S^{\min}$.
\begin{lemma}\label{lemma3}$^{\ast}$
For any minimum selective subset $S^{min}$ of $G$, we have 
\begin{itemize}
    \item No vertex $v\in B_{i}\setminus B_{i,3}$ belongs to $S^{min}$.
    \item $S^{min}\subseteq \bigcup_{i=1}^{k}B_{i,3}$.
\end{itemize}
\end{lemma}

Lemma~\textcolor{blue}{\ref{lemma2}} ensures that for each vertex $v \in B_{i,1}$, either $v \in S^{min}$ or at least one of its adjacent vertices in $B_{i,3}$ belongs to $S^{min}$. Importantly, the choice of including $v$ itself or one of its adjacent vertices from $B_{i,3}$ in $S^{min}$ does not affect the selection of vertices in other blocks. Combined with Lemma~\textcolor{blue}{\ref{lemma3}}, which establishes that $S^{\min} \subseteq \bigcup_{i=1}^{k} B_{i,3}$, we obtain the following remark:
\begin{remark}\label{obs3}
The blocks are independent in constructing a selective subset; that is, the selection of vertices in one block does not constrain the selection in other blocks.
\end{remark}
We now apply an appropriate algorithm to each block separately due to Remark~\textcolor{blue}{\ref{obs3}}. To do this, we first define a selective subset for each block as follows. 
\begin{definition}[Selective Subset of $B_{i,1}$]
A selective subset of $B_{i,1}$, denoted by $S(B_{i,1})$, is a subset of $B_{i,3}$ such that for every vertex $v \in B_{i,1}$, either $v \in S(B_{i,1})$ or $\n(v,B_{i,3}) \cap S(B_{i,1}) \neq \emptyset$. By Lemmas~\textcolor{blue}{\ref{lemma2}} and~\textcolor{blue}{\ref{lemma3}}, it follows that $S(B_{i,1}) \subseteq \n[B_{i,1},B_{i,3}] \subseteq B_{i,3}$.
\end{definition}
\begin{theorem}\label{th3}$^{\ast}$
Let $G$ be a connected graph with blocks $B_{1},\dots,B_{k}$, and for each block $B_i$ let $S(B_{i,1})$ be any selective subset of $B_{i,1}$.
Let $S=\bigcup_{i=1}^{k} S(B_{i,1})$, then $S$ is a selective subset of $G$. Moreover, if each $S(B_{i,1})$ satisfies 
\(
|S(B_{i,1})| \le (1+\epsilon)\,|S^{\min}(B_{i,1})|,
\)
then
$|S| \le (1+\epsilon)\,|S^{\min}|$,
where $S^{\min}$ denotes a minimum selective subset of $G$.
\end{theorem}

\subsection{Finding Local Selective Subsets}
We now establish a bound on a local solution using \emph{2-distance subsets} for our problem and then merge all local solutions to obtain the desired solution.
\begin{definition}[2-distance Subsets]
A collection of subsets of the vertices in $B_{i,1}$, denoted as $D_i=\{D_{i,1}^{1},\dots, D_{i,1}^{t_i}\}$, is called a collection of \emph{2-distance subsets} if the following properties hold (see example in Figure~\textcolor{blue}{\ref{bubai1}}):
\begin{itemize}
    \item $D_{i,1}^{j}\subseteq B_{i,1}$ for all $1\leq j\leq t_i$.
    \item The subgraph $G[D_{i,1}^{j}]$ is connected in the induced subgraph $G[B_{i,3}]$, i.e., the induced subgraph $G[D_{i,1}^{j}]$ may not be connected itself, but any two vertices in $D_{i,1}^{j}$ must have a path between them in $G[B_{i,3}]$.
    \item The subsets are pairwise at a distance greater than two in $G[B_{i,3}]$, i.e., $\dist (D_{i,1}^j, D_{i,1}^l)>2$ in the subgraph $G[B_{i,3}]$ when $j\neq l$. 
\end{itemize}
\end{definition}
\begin{definition}[Local Selective Subset]
A \emph{local} selective subset of $D_{i,1}^j$, denoted by $S(D_{i,1}^j)$, is a subset of $B_{i,3}$ such that for every vertex $v \in D_{i,1}^j$, either $v \in S(D_{i,1}^j)$ or $\n(v,B_{i,3}) \cap S(D_{i,1}^j) \neq \emptyset$. By Lemmas~\textcolor{blue}{\ref{lemma2}} and~\textcolor{blue}{\ref{lemma3}}, it follows that
$S(D_{i,1}^j) \subseteq \n[D_{i,1}^j,B_{i,3}] \subseteq B_{i,3}$.
\end{definition}

\begin{lemma}\label{lemma4}$^{\ast}$
For any $j\neq l$, the following holds:
\begin{itemize}
    \item $\n [D_{i,1}^j, B_{i,3}]\cap  \n [D_{i,1}^l,B_{i,3}]=\emptyset$.
    \item $S^{min}(D_{i,1}^j) \cap S^{min}(D_{i,1}^l)= \emptyset$.
    \item $\big(S^{min}\cap S^{min}(D_{i,1}^{j})\big)\cap\big(S^{min}\cap S^{min}(D_{i,1}^{l})\big)=\emptyset.$
\end{itemize}
\end{lemma}
Lemma~\textcolor{blue}{\ref{lemma4}} implies that the solutions $S^{min}(D_{i,1}^{j})$ and $S^{min}(D_{i,1}^{l})$ do not share a common vertex in $S^{min}$ for $j\neq l$.
\begin{lemma}\label{manna1}$^{\ast}$
$S^{min}\cap \n [D_{i,1}^j, B_{i,3}]$ is a local selective subset of $D_{i,1}^j$.
\end{lemma}
\begin{lemma}\label{lemma5}$^{\ast}$
For any collection of 2-distance subsets $D_i=\{D_{i,1}^1,$$\dots,$ $ D_{i,1}^{t_i}\}$,  where $1\leq i\leq k$ in the graph $G$; we have: $\sum_{i=1}^{k}\sum_{j=1}^{t_i}\lvert S^{min}(D_{i,1}^{j})\rvert \leq \lvert S^{min}\rvert.$
\end{lemma}
Lemma~\textcolor{blue}{\ref{lemma5}} shows that 2-distance subsets yield a lower bound on the size of a minimum selective subset. However, the set $\sum_{i=1}^{k}\sum_{j=1}^{t_i} S^{min}(D_{i,1}^{j})$ need not form a selective subset of the entire graph $G$. To construct a selective subset for $G$, we enlarge each $D_{i,1}^j$ to a corresponding set $E_{i,1}^j$ that remains locally bounded while still providing a valid local solution. This enlargement allows us to approximate a selective subset of $G$.
\begin{theorem}\label{th4}$^{\ast}$
Let $D_{i}=\{D_{i,1}^1,\dots,D_{i,1}^{t_i}\}$ be a collection of 2-distance subsets, and $\{E_{i,1}^1,\dots,E_{i,1}^{t_i}\}$ be the corresponding collection of subsets of $B_{i,1}$ such that $D_{i,1}^{j}\subseteq E_{i,1}^{j}$ for all $1\leq i\leq k$ and $1\leq j\leq t_i$. 

If there exists a bound $\delta\geq 1$ such that $\lvert S^{min}(E_{i,1}^{j})\rvert\leq \delta \cdot \lvert S^{min}(D_{i,1}^j)\rvert$ 
for all $1\leq i\leq k$ and $1\leq j\leq t_i$, and if $\bigcup_{i=1}^{k}\bigcup_{j=1}^{t_i}S^{min}(E_{i,1}^j)$ forms a selective subset of $G$, then $\bigcup_{i=1}^{k}\bigcup_{j=1}^{t_i}S^{min}(E_{i,1}^j)$ is a $\delta$-approximation of a minimum selective subset of $G$.
\end{theorem}
\subsection{Finding a Global Selective Subset}
For $r=0,1,\dots$, we recursively define the $r$-th neighborhood of any vertex $v \in B_{i,1}$ in $B_{i,3}$ by
\[
\n_i^{r}[v,B_{i,3}] = \n[\n_i^{r-1}[v,B_{i,3}], B_{i,3}],
\]
with
\[
\n_i^{0}[v,B_{i,3}] = \{v\}, \quad \n_i^{1}[v,B_{i,3}] = \n[v,B_{i,3}].
\]
Since $\n _i^{r}[v, B_{i,3}] \subseteq B_{i,3}$, we partition $\n_i^{r}[v,B_{i,3}]$ into $X_i^{r} \subseteq B_{i,1}$ and $Y_i^{r}\subseteq B_{i,2}$. We will later use $X_i^{r}$ and $Y_i^{r}$ in our algorithm.

As $\delta \geq 1$, we assume that $\delta :=(1+\epsilon)$. The key idea is to determine the neighborhood of a vertex in $B_{i,3}$ and then progressively expand this neighborhood until we obtain sets $D_{i,1}^{j}$ and $E_{i,1}^{j}$ (where $E_{i,1}^{j}\supseteq D_{i,1}^{j}$) that satisfy Theorem~\textcolor{blue}{\ref{th4}}. Once this is achieved, we remove the current neighborhood and repeat the process for the remaining graph. Note that \( E_{i,1}^{j} \) is not a 2-distance subset, but \( D_{i,1}^{j} \) is. The complete procedure is described below (see the pseudocode in Algorithm~\textcolor{blue}{\ref{alg:01}$^{\ast}$}):
\begin{itemize}
    \item Initially, set $i \leftarrow 1$.
    \item \textbf{Stage 1:} Initialize $j \leftarrow 1$, and set $B_{i,1}^{j} \leftarrow B_{i,1}$, $B_{i,2}^{j} \leftarrow B_{i,2}$, and $B_{i,3}^{j} \leftarrow B_{i,3}$. 
    \item \textbf{Stage 2:} Choose an arbitrary vertex $v_i^{j}$ from $B_{i,1}^{j}$.
    \item For $r = 0,1,\dots$, consider the $r$-th neighborhood $\n _{i,j}^{r} [v_i^{j}, B_{i,3}^{j}]$.  
    Starting with $\n _{i,j}^{0} [v_i^{j},B_{i,3}^{j}]$ and compute the minimum selective subset while inequality~(\textcolor{blue}{\ref{eq1}}) holds.  
    \begin{equation}
        \lvert S^{min}(X_{i,j}^{r+2})\rvert > \delta \cdot \lvert S^{min}(X_{i,j}^{r})\rvert
        \label{eq1}
    \end{equation}
    Here $\n_{i,j}^{r}$ (rather than $\n_{i}^{r}$) denotes the $r$-th neighborhood used to compute $D_{i,1}^{j}$ and $E_{i,1}^{j}$ from $B_{i,3}$. The same convention applies to $X_{i,j}^{r}$ and $Y_{i,j}^{r}$.
    \item Let $\overline{r_{i,j}}$ be the smallest $r$ for which inequality~(\textcolor{blue}{\ref{eq1}}) is violated, i.e.,  
    \[
        \lvert S^{min}(X_{i,j}^{\overline{r_{i,j}}+2} )\rvert \leq \delta \cdot \lvert S^{min}(X_{i,j}^{\overline{r_{i,j}}})\rvert.
    \]
    \item Update the sets as follows:  
    \begin{align*}
        D_{i,1}^{j} &\leftarrow X_{i,j}^{\overline{r_{i,j}}},\\
        E_{i,1}^{j} &\leftarrow X_{i,j}^{\overline{r_{i,j}}+2},\\
        B_{i,3}^{j+1} &\leftarrow B_{i,3}^{j} \setminus \n _{i,j} ^{\overline{r_{i,j}}+2}[v_i^{j},B_{i,3}^{j}], \\  
        B_{i,1}^{j+1} &\leftarrow B_{i,1}^{j} \setminus X _{i,j} ^{\overline{r_{i,j}}+2}, \\  
        B_{i,2}^{j+1} &\leftarrow B_{i,2}^{j} \setminus Y _{i,j} ^{\overline{r_{i,j}}+2},\\  
        j &\leftarrow j+1.
    \end{align*}
    \item Repeat the process from \textbf{Stage 2} until $B_{i,3}^{j} = \emptyset$.
    \item Once $B_{i,3}^{j}$ becomes empty, set $i \leftarrow i+1$ and repeat the process from \textbf{Stage 1} until $i=k+1$.
\end{itemize}
Suppose the sets \( D_{i,1}^{1}, $ $ D_{i,1}^{2}, $ $ \dots, D_{i,1}^{t_i} \) and \( E_{i,1}^{1}, $ $E_{i,1}^{2}, $ $\dots, E_{i,1}^{t_i} \) are returned from the above algorithm for \( 1\leq i\leq k \). We establish the following lemmas.
\begin{lemma}\label{lemma7}$^{\ast}$
The sets $\{D_{i,1}^{1}, D_{i,1}^{2}, \dots, D_{i,1}^{t_i}\}$, where \( 1\leq i\leq k \), obtained from the above algorithm, form a collection of 2-distance subsets.
\end{lemma}
\begin{lemma}\label{lemma8}$^{\ast}$
For the collection of sets \( \{E_{i,1}^{1},\dots, E_{i,1}^{t_i}\} \) obtained from the above algorithm, the union $S=\bigcup_{i=1}^{k}\bigcup_{j=1}^{t_i} S^{min}(E_{i,1}^{j})$
forms a selective subset of \( G \).
\end{lemma}
Combining Lemmas~\textcolor{blue}{\ref{lemma7}} and~\textcolor{blue}{\ref{lemma8}} with Theorem~\textcolor{blue}{\ref{th4}}, we obtain the following theorem directly:
\begin{theorem}\label{th5}
The above algorithm produces a selective subset $\bigcup_{i=1}^{k}$ $\bigcup_{j=1}^{t_i} $ $ S^{min}($ $E_{i,1}^{j})$
of size at most \( (1+\epsilon) \) times the size of the minimum selective subset of \( G \).
\end{theorem}
Until now, we have considered general graphs rather than unit disk graphs. Thus, Theorem~\ref{th5} holds for general graphs.
\subsection{Finding $S^{min}(E_{i,1}^{j})$ on Unit Disk Graphs}
The only remaining task is to compute $S^{min}(E_{i,1}^{j})$ in time $n ^ {f(\epsilon)}$ on unit disk graphs. We assume that \( F_{i,j} = \n _{i,j} ^{\overline{r_{i,j}}+2}[v_i^{j},B_{i,3}^{j}] \). According to the above algorithm,  
\[
F_{i,j} = X_{i,j}^{\overline{r_{i,j}}+2} \cup Y_{i,j}^{\overline{r_{i,j}}+2} = E_{i,1}^{j} \cup Y_{i,j}^{\overline{r_{i,j}}+2}.
\]
We first show that the size of \( S^{min}(E_{i,1}^{j}) \) is at most the size of the \textsc{Maximum Independent Set} of \( F_{i,j} \). This provides a bound on \( E_{i,1}^{j} \). The \textsc{Maximum Independent Set} of a graph \( G \) is the largest set of vertices such that no two vertices in the set are adjacent. One might think that \( S^{min}(E_{i,1}^{j}) \) is the same as a \textsc{Minimum Dominating Set} of \( F_{i,j} \). However, by Theorem~\textcolor{blue}{\ref{th3}}, some vertices in \( B_{i,3} \setminus B_{i,1} \) may have no adjacent vertex (including themselves) in \( S(B_{i,1}) \). Therefore, with this bound, it suffices to compute each \( S^{min}(E_{i,1}^{j}) \).
\begin{lemma}\label{lemma10}$^{\ast}$
The size of \( S^{min}(E_{i,1}^{j}) \) is at most the size of the maximum independent set of \( F_{i,j} \).
\end{lemma}
Now, we apply the method described in~\textcolor{blue}{\cite{independent}} to find a \textsc{Maximum Independent Set} in \( F_{i,j} \) for unit disk graphs.
\begin{lemma}\label{lemma11}$^{\ast}$
For any unit disk graph \( U \) and an independent set \( I^{r} \subseteq \n_{i,j} ^{r}[v_i^{j},$ $ B_{i,3}^{j}] \), we have $\lvert I^{r} \rvert=(2r+1)^2 = \mathcal{O}(r^2).$
\end{lemma}
Using Lemmas~\textcolor{blue}{\ref{lemma10}} and~\textcolor{blue}{\ref{lemma11}}, we derive the following theorem directly:
\begin{theorem}\label{th6}
The size of the minimum selective subset (local) of $E_{i,1}^{j}$ satisfies
\[
\lvert S^{\min}(E_{i,1}^{j}) \rvert = \mathcal{O}(r^2).
\]
\end{theorem}
\begin{lemma}\label{lemma20}$^{\ast}$
There exists a constant \( d(\delta) \), depending on \( \delta = (1+\epsilon) \), such that \( \overline{r_{i,j}} \leq d(\delta) \). The running time of our algorithm to compute a \ptas is \( \mathcal{O}(n^{d^2}) \), where \( |V(U)| = n \), \( d(\epsilon) = \mathcal{O}\Big(\frac{1}{\epsilon^2} \log \frac{1}{\epsilon}\Big) \), and \( 0 < \epsilon < \frac{1}{10} \).
\end{lemma}

\section{APX-hardness of \mss on Circle Graphs}\label{apxhard}
The vertex set of a circle graph is a set of chords of a given circle, and if two chords intersect, the corresponding vertices share an edge. We obtain a ``gap-preserving'' reduction from the \maxsat problem to a circle graph using the \textsc{Minimum Dominating Set} (\mds) problem on circle graphs. The \maxsat problem is as follows~\textcolor{blue}{\cite{10.5555/1965254}}:

We are given a set of $n$ variables $\mathcal{X}=\{x_1,\dots,x_n\}$ and $m$ clauses $\mathcal{C} =\{c_1,\dots,c_m\}$ such that each clause has at most $3$ literals and each variable occurs in at most 8 clauses. The objective is to find a truth-assignment of the variables in $\mathcal{X}$ that maximizes the number of clauses in $\mathcal{C}$ satisfied. \maxsat\ is \apxh, and the \mds\ problem on circle graphs is also \apxh\ \textcolor{blue}{\cite{circle}}.

Consider an instance $\phi$ of \maxsat. Let $\text{SAT}(\phi)$ represent the maximum fraction of clauses in $\phi$ that can be satisfied. For a given graph $G$, let $\gamma(G)$ denote the cardinality of its \textsc{Minimum Dominating Set}. The paper~\textcolor{blue}{\cite{circle}} reduces an instance of \maxsat to a circle  
graph $G$ such that $\lvert V(G) \rvert = m + 56n + 4$ and states the following theorem:
\begin{theorem}\label{theoremcircle}
A polynomial-time reduction transforms an instance $\phi$ of \maxsat, consisting of $n$ variables and $m$ clauses, into a circle graph $G$ such that  
\[
\text{SAT}(\phi) = 1 \implies \gamma(G) \leq 16n + 2,
\]
\[
\text{SAT}(\phi) < \alpha \implies \gamma(G) > 16n + 2 + \frac{(1 - \alpha)m}{8}, \quad \text{for any } 0 < \alpha < 1.
\]
\end{theorem}
We now reduce the graph $G$ into a graph $H$ as follows:
\begin{figure}[ht]
\includegraphics[width=9cm]{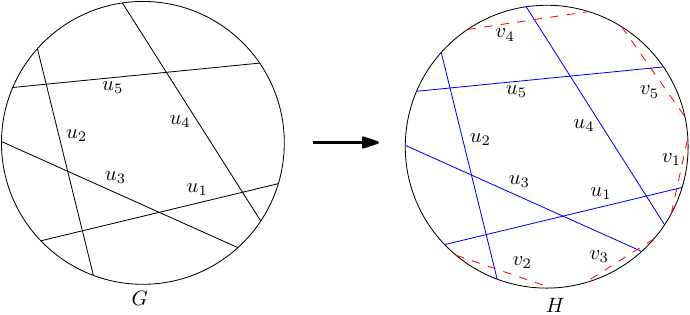}
\centering
\caption{Colors: \emph{blue $\equiv$ normal chord}, and \emph{red $\equiv$ dotted chord}. An example of the reduction: each chord of \( G \) is colored \textit{blue} in \( H \). Each \textit{red} chord in \( H \) is adjacent only to a chord of \textit{blue} color.}\label{bubai3} 
\end{figure}

\textbf{Reduction.} Define $V(H) = V(G) \cup V(X)$ and $E(H) = E(G) \cup E(X)$, where $X$ is a subgraph with vertex set $V(X)$ and edge set $E(X)$ are constructed as follows (see Figure~\textcolor{blue}{\ref{bubai3}}). Initially, set $V(X) := \emptyset$ and $E(X) := \emptyset$. Assign color $0$ to all chords in $V(G)$, i.e., $C(V(G)) = \{0\}$. For each chord $u_i \in V(G)$, introduce a corresponding chord $v_i \in V(X)$ such that $C(v_i) = 1$. Position $v_i$ so that it intersects only $u_i$ (meaning $v_i$ has degree 1 in $H$). These newly created edges are then included in $E(X)$. As a result, we obtain $\lvert V(H) \rvert = 2m+112n+8$, $\lvert V(X)\rvert =m+56n+4$ and $\lvert C(V(H)) \rvert = 2$.
\begin{lemma}{\label{lemmaapx}}$^{\ast}$
$G$ has a dominating set of size $t$ if and only if $H$ has a selective subset of size $m+56n+4+t$.
\end{lemma}
Let \( \lvert S^{min}(H)\rvert \) denote the cardinality of the minimum selective subset of \( H \). We now show the ``gap-preserving'' reduction from \maxsat to the graph \( H \) in the following theorem:
\begin{theorem}\label{bubailast}
The \mss\ problem is \apxh\ in circle graphs.
\end{theorem}

\begin{remark}
Our \apxh result holds not only for circle graphs with two colors; indeed, any color (other than blue) can replace the red chords adjacent to the blue ones.
\end{remark} 
\section{Conclusion}
It is still open regarding whether \mss\ problem is \npc\ on unit disk graphs when \( c \) is constant. This raises the open question of whether \fpt\ is possible when the number of colors \( c \) is treated as a parameter. Another direction to explore is whether \fpt\ is possible when the number of blocks \( k \) is treated as a parameter. Additionally, since \mss is \apxh on circle graphs, it is natural to ask whether a constant-factor approximation, or a \((2 + \epsilon)\)-approximation is possible for circle graphs, in contrast to a \((1+\epsilon)\)-approximation. \fpt for \mss problem with respect to parameters \( c \) or \( k \) also remains an open question in the context of circle graphs.

Moreover, exploring the complexity of \mss{} on additional graph classes—such as circular-arc graphs, chordal graphs, and permutation graphs—could uncover new tractable cases or reveal deeper structural insights. These graph families often arise in scheduling, bioinformatics, and network analysis, and understanding \mss{} within these domains may have practical implications as well.
\section*{Acknowledgments}
I sincerely thank Bodhayan Roy and Aritra Banik for their support and guidance. Also, the motivation for this work arose from discussions on the \mcs paper with Aritra Banik, facilitated by the ACM India Anveshan Setu Fellowship program in 2024. I also thank the anonymous reviewers for their thoughtful comments and suggestions, which helped enhance the quality and presentation of the paper.
\bibliographystyle{unsrt} 
\bibliography{references}  
\newpage
\appendix
\section{Proof of the Lemma~\ref{lemma2}}
\begin{proof}
Suppose $\n [v,B_{i,3}]\cap S=\emptyset$. Let $v\in V_{\ell}$ for some $1\leq \ell \leq c$. Since each vertex of $B_{i,1}$ must have at least one adjacent vertex of a different color, there must exist a vertex $w\notin B_{i}$ such that $\dist (v,w)=1$ and $C(w)\neq C(v)$. If no vertex of \( \n[v, B_{i,3}] \) is in \( S \), then \( v \) has \( w \) as its only nearest neighbor in \( S \cup (V(G) \setminus V_{\ell}) \), and \( C(v) \neq C(w) \), contradicting the assumption that \( S \) is a selective subset.\qed
\end{proof}
\section{Proof of the Lemma~\ref{bubailemma6}}
\begin{proof}
Let $D$ be a \textsc{Dominating Set} of $G$ of size $s$. Define $S = D \cup \{z\}$; then $\lvert S\rvert = s+1$. We claim that $S$ is a selective subset of $G'$. If not, then there exists a vertex $u\in V(G)\setminus S$ such that $\dist (u, S)> \dist(u,z)$ which implies that $D$ is not a \textsc{Dominating Set} of $G$.

Conversely, suppose $S$ is a selective subset of size $s+1$ in the graph $G'$. Due to Lemma~\textcolor{blue}{\ref{lemma2}}, $z\in S$. Let $D=S\setminus \{z\}$. We claim that $D$ is a \textsc{Dominating Set} of $G$. If not, then there exists a vertex $u\in V(G)$ such that \( \n [u, V(G)] \cap D = \emptyset \), which would imply $\dist (u,S) > \dist (u,z)$, contradicting the fact that $S$ is a selective subset.
Hence our lemma is proved.\qed
\end{proof}
\section{Proof of the Theorem~\ref{th1}}
\begin{proof}
We recall that the \textsc{Minimum Dominating Set} problem is known to be $\apxh$ on general graphs; that is, there exists a constant $\delta > 0$ such that no polynomial-time algorithm can approximate the solution within a factor of $\delta \cdot \log n$ unless $\npp = \np$. 

By Lemma~\textcolor{blue}{\ref{bubailemma6}}, there is an approximation-preserving reduction from the \textsc{Minimum Dominating Set} problem to the \mss problem. In particular, given any instance of \textsc{Minimum Dominating Set}, one can construct in polynomial time an instance of \mss such that any approximation algorithm for \mss would translate to an approximation algorithm of comparable quality for \textsc{Minimum Dominating Set}. 

Therefore, since \textsc{Minimum Dominating Set} is $\apxh$ on general graphs, it follows immediately from the reduction that the \mss problem inherits the same hardness of approximation. Hence, the \mss problem is $\lapxh$ on general graphs.\qed
\end{proof}
\section{Proof of the Lemma~\ref{npclemma}}
\begin{proof}
Let $D$ be a dominating set of $U$ of size $t$. Each vertex in $V(X)$ forms a block, since each vertex has a unique color. By Lemma~\textcolor{blue}{\ref{lemma2}}, every selective subset $S$ must include all vertices in $V(X)$, i.e., $V(X) \subseteq S$. Therefore, to construct a selective subset of $U'$, we take all vertices in $V(X)$ along with the corresponding vertices of $D$. Let $S = V(X) \cup D$. Since every vertex in $V(U)$ is either in $D$ or adjacent to a vertex in $D$, and since each vertex in $V(X)$ is already in $S$, we conclude that $S$ is a selective subset of size $nm+t$. This holds because for each vertex $v \in V(U')\setminus S$, $\dist(v,S) = 1 = \dist(v,V(X))$.

Conversely, suppose $S$ is a selective subset of $U'$ of size $nm+t$. By Lemma~\textcolor{blue}{\ref{lemma2}}, $S$ must contain all vertices of $V(X)$. Define \( D' = S \setminus V(X) \). We claim that \( D' \) is a dominating set of $U$ with \( \lvert D' \rvert = t \). For every vertex $u \in V(U)$, $S$ ensures that either $u \in S$ or $u$ has an adjacent vertex in $S$ of the same color. This holds since $\dist (u,V(X))=1$ and $C(u)\notin C(V(X))$. Therefore, the corresponding vertices in $D'$ must satisfy the same conditions. Thus, $D'$ is a dominating set of $U$ of size $t$.\qed
\end{proof}
\section{Proof of the Theorem~\ref{th2}}
\begin{proof}
    It is straightforward to see that \mss belongs to \np, since a given solution can be verified in polynomial time. By Lemma~\textcolor{blue}{\ref{npclemma}}, we established a reduction from the $\npc$ \mds problem in \udg~\textcolor{blue}{\cite{Clark1990}}. Hence, \mss is $\npc$ on unit disk graphs.\qed
\end{proof}
\section{Proof of the Lemma~\ref{lemma3}}
\begin{proof}
Suppose that $v\in B_i\setminus B_{i,3}$ and $v\notin S^{min}$. Let $v\in V_{\ell}$ and $w\in \NN (v,B_{i,1})$. By Lemma~\textcolor{blue}{\ref{lemma2}}, we have $\n [w,B_{i,3}]\cap S^{min}\neq \emptyset$. Since $v\notin B_{i,3}$, $\n [w,B_{i,3}]\cap (S^{min}\setminus \{v\})\neq \emptyset$. Therefore for any $j\neq i$, $\dist (v,B_j) \geq \dist (v,w)+\dist \big(w, \n (w,B_{i,3})\big)$.
    
This implies that $v$ must have a nearest vertex in $S^{min}\setminus \{v\}\cup (V(G)\setminus V_{\ell})$ with the same color as $C(v)$. Hence, no vertex from $B_{i}\setminus B_{i,3}$ is in $S^{min}$.

Since no vertex \( v \in B_{i} \setminus B_{i,3} \) is in \( S^{min} \), it follows that \( S^{min} \subseteq \bigcup_{i=1}^{c}B_{i,3} \).\qed
\end{proof}
\section{Proof of the Theorem~\ref{th3}}
\begin{proof}
Let $C(B_i) = \ell$. By Lemma~\textcolor{blue}{\ref{lemma3}}, every minimum selective subset is contained in $\bigcup_{i=1}^{k} B_{i,3}$.  
Hence, it remains to show that every vertex in $B_{i,3} \setminus S(B_{i,1})$ has a nearest neighbor of its own color in $S(B_{i,1}) \cup (V(G)\setminus V_\ell)$.

Consider a vertex $w \in B_{i,3} \setminus S(B_{i,1})$. We distinguish two cases:

\textbf{Case 1.} $w \in B_{i,1}$. By Lemma~\textcolor{blue}{\ref{lemma2}}, at least one adjacent vertex of $w$ in $B_{i,3}$ must belong to $S(B_{i,1})$. So, $w$ has a nearest neighbor in $S(B_{i,1})\cup (V(G)\setminus V_{\ell})$ of its own color.

\textbf{Case 2.} $w \in B_{i,2}$. Let $u \in B_{i,1}$ be an adjacent vertex of $w$. Since $w \notin S(B_{i,1})$ and $u\in B_{i,1}$, by Lemma~\textcolor{blue}{\ref{lemma2}}, it follows that $(\n[u,B_{i,3}] \setminus \{w\}) \cap S(B_{i,1}) \neq \emptyset$. Because $\dist(w,u) = 1$, we have $\dist(w,B_j) \;\geq\; \dist(w,u) + \dist\big(u, (\n(u,B_{i,3}) \setminus \{w\})\big)$
   for all $j \neq i$.  
   Thus, $w$ has a nearest neighbor of its own color in $(\n[u,B_{i,3}] \setminus \{w\}) \cap S(B_{i,1})$.

Therefore, $S(B_{i,1})$ is also a selective subset for $B_i$. Consequently, constructing $S = \bigcup_{i=1}^{k} S(B_{i,1})$
with size at most a $(1+\epsilon)$-factor of the global optimum is sufficient.\qed
\end{proof}
\section{Proof of the Lemma~\ref{lemma4}}
\begin{proof}
Since $\dist (D_{i,1}^j, D_{i,1}^l)>2$, we have $\dist (\n [D_{i,1}^j,$ $B_{i,3}], $ $ \n [D_{i,1}^l,$ $B_{i,3}])\geq 1$, which implies that their neighborhoods are disjoint. Hence, the first claim holds.

Since $S^{min}(D_{i,1}^j)$ $\subseteq \n [D_{i,1}^j,B_{i,3}]$ and $S^{min}(D_{i,1}^l)$ $\subseteq \n [D_{i,1}^l,$ $B_{i,3}]$, the second claim follows immediately.

Finally, since $S^{min}(D_{i,1}^j) \cap S^{min}(D_{i,1}^l)= \emptyset$, the third claim also holds. \qed 
\end{proof}
\section{Proof of the Lemma~\ref{manna1}}
\begin{proof}
From Lemma~\textcolor{blue}{\ref{lemma3}}, every vertex in \( S^{\min}\) lies within the neighborhood \( \bigcup_{i=1}^{k} \n[B_{i,1}, B_{i,3}] \). Lemma~\textcolor{blue}{\ref{lemma2}} further states that for each vertex \( v \in B_{i,1} \), either \( v \in S^{\min} \), or at least one of its adjacent vertices \( u \in B_{i,3} \) must be included in \( S^{\min} \).

Moreover, by the proof of Theorem~\textcolor{blue}{\ref{th3}}, for any vertex \( x \in \bigcup_{i=1}^{k} \n[B_{i,1}, B_{i,3}] \setminus S^{\min} \), if \( x \in V_{\ell} \), then there exists a nearest vertex \( y \in S^{\min} \cup (V(G) \setminus V_{\ell}) \) such that \( C(x) = C(y) \). Therefore, $x$ satisfies the definition of the selective subset problem.

Therefore, for every 2-distance subset \( D_{i,1}^{j} \), the set \( S^{\min} \cap \n[D_{i,1}^{j}, B_{i,3}] \) forms a local selective subset for \( D_{i,1}^{j} \), since every vertex \( v \in D_{i,1}^{j}\subseteq B_{i,1} \) is either in the set \( S^{\min} \cap \n[D_{i,1}^{j}, B_{i,3}] \) or has at least one adjacent vertex in \( \n[D_{i,1}^{j}, B_{i,3}] \).\qed
\end{proof}
\section{Proof of the Lemma~\ref{lemma5}}
\begin{proof}
By Lemma~\textcolor{blue}{\ref{lemma4}}, since the neighborhoods of any two 2-distance subsets are disjoint, the corresponding local minimum selective subsets are disjoint. Therefore, each $S^{min}(D_{i,1}^j)$ contributes uniquely to $S^{min}$.

Moreover, by Lemma~\textcolor{blue}{\ref{manna1}}, $S^{min}\cap \n [D_{i,1}^j, B_{i,3}]$ must be a local selective subset of $D_{i,1}^j$. Also, $S^{min} (D_{i,1}^j)$ is the minimum selective subset for $D_{i,1}^j$. Therefore, 
$$\lvert S^{min} (D_{i,1}^j) \rvert\leq \lvert S^{min}\cap \n [D_{i,1}^j, B_{i,3}]\rvert.$$
Summing over all $i$ and $j$, we get:
$$ \sum_{i=1}^{k}\sum_{j=1}^{t_i}\lvert S^{min}(D_{i,1}^{j})\rvert\leq \sum_{i=1}^{k}\sum_{j=1}^{t_i}\lvert S^{min}\cap \n [D_{i,1}^j, B_{i,3}]\rvert \leq \lvert S^{min}\rvert.$$\qed
\end{proof}
\section{Algorithm~\ref{alg:01}}
\begin{algorithm}[H]\label{alg:01}
\caption{Algorithm for Computing $D_{i,1}^{j}$ and $E_{i,1}^{j}$}
\KwIn{Graph \( G = (V(G), E(G)) \), initial sets \( B_{i,1}, B_{i,2}, B_{i,3} \) for $1\leq i\leq k$, and parameter \( \delta \).}
\KwOut{Sets \(D_{i,1}^{j}\) and \(E_{i,1}^{j}\) for all \(i,j\).}

Initialize \( i \leftarrow 1 \)\;

\While{\( i \leq k \)}{
    Initialize \( j \leftarrow 1 \)\;
    Set \( B_{i,1}^{j} \leftarrow B_{i,1} \), \( B_{i,2}^{j} \leftarrow B_{i,2} \), \( B_{i,3}^{j} \leftarrow B_{i,3} \)\;

    \While{\( B_{i,3}^{j} \neq \emptyset \)}{
        Select an arbitrary vertex \( v_i^{j} \in B_{i,1}^{j} \)\;

        \For{\( r = 0,1,2,\dots \)}{
            Compute \( \n _{i,j}^{r} [v_{i}^{j}, B_{i,3}^{j}] \)\;
            Compute \( S^{min}(X_{i,j}^{r+2}) \)\;
            
            \If{\( \lvert S^{min}(X_{i,j}^{r+2})\rvert \leq \delta \cdot \lvert S^{min}(X_{i,j}^{r})\rvert \)}{
                Break\;
            }
        }

        Set \( \overline{r_{i,j}} \) as the smallest \( r \) for which the condition is violated\;

        Update sets:\;
        \( D_{i,1}^{j} \leftarrow X_{i,j}^{\overline{r_{i,j}}} \)\;
        \( E_{i,1}^{j} \leftarrow X_{i,j}^{\overline{r_{i,j}}+2} \)\;
        \Return{Selective subset for \( D_{i,j}, E_{i,j} \)}
        \( B_{i,3}^{j+1} \leftarrow B_{i,3}^{j} \setminus \n _{i,j}^{\overline{r_{i,j}}+2} [v_{i}^{j}, B_{i,3}^{j}] \)\;
        \( B_{i,1}^{j+1} \leftarrow B_{i,1}^{j} \setminus X _{i,j}^{\overline{r_{i,j}}+2} \)\;
        \( B_{i,2}^{j+1} \leftarrow B_{i,2}^{j} \setminus Y _{i,j}^{\overline{r_{i,j}}+2} \)\;
         
        Increment \( j \leftarrow j+1 \)\;
    }

    Increment \( i \leftarrow i+1 \)\;
}
\end{algorithm}
\section{Proof of the Lemma~\ref{lemma7}}
\begin{proof}
We prove this by mathematical induction on \( j = 1, \dots, t_i \) for each \( i \). 

For the base case \( j=1 \), we have 
\[
B_{i,3}^{2} = B_{i,3}^{1} \setminus \n_{i,1}^{\overline{r_{i,j}}+2}[v_i^{1}, B_{i,3}^{1}].
\]
Since 
\[
\n [\n [D_{i,1}^{1}, B_{i,3}^{1}]] = \n [\n [\n_{i,1}^{\overline{r_{i,j}}}[v_i^{1}, B_{i,3}^{1}]]] = \n_{i,1}^{\overline{r_{i,j}}+2}[v_i^{1}, B_{i,3}^{1}],
\]
it follows that \( \{D_{i,1}^{1}, B_{i,1}^{2}\} \) form a 2-distance subset because \( \dist (D_{i,1}^{1}, B_{i,1}^{2}) > 2 \) and \( D_{i,1}^{1}, B_{i,1}^{2} \subseteq B_{i,1} ^{1}\).

Now, as the induction hypothesis, assume that \( \{D_{i,1}^{1} $ $,\dots,$ $ D_{i,1}^{j-1}, $ $ B_{i,1}^{j}\} \) forms a 2-distance subset. We need to show that \( \{D_{i,1}^{1},\dots, D_{i,1}^{j}, B_{i,1}^{j+1}\} \) is also a 2-distance subset. 

Since 
\[
B_{i,3}^{j+1} =B_{i,3}^{j} \setminus \n_{i,1}^{\overline{r_{i,j}}+2}[v_i^{j}, B_{i,3}^{j}]= B_{i,3}^{j} \setminus \n [\n [D_{i,1}^{j}, B_{i,3}^{j}]],
\]
this completes the proof.\qed
\end{proof}
\section{Proof of the Theorem~\ref{th4}}
\begin{proof}
By summing over all indices, we obtain:
\begin{equation*}
\begin{split}
\big\lvert \bigcup_{i=1}^{k} \bigcup_{j=1}^{t_i} S^{min}(E_{i,1}^j) \big\rvert 
&\le \sum_{i=1}^{k}\sum_{j=1}^{t_i} \lvert S^{min}(E_{i,1}^j) \rvert \\
&\le \delta \sum_{i=1}^{k}\sum_{j=1}^{t_i} \lvert S^{min}(D_{i,1}^j) \rvert \\
&\le \delta \lvert S^{min} \rvert,
\end{split}
\end{equation*}
where the final step follows from Lemma~\textcolor{blue}{\ref{lemma5}}.\qed
\end{proof}
\section{Proof of the Lemma~\ref{lemma8}}
\begin{proof}
Since \( B_{i,1}^{j+1} = B_{i,1}^{j} \setminus E_{i,1}^{j} \), where \( E_{i,1}^{j} = X_{i,j}^{\overline{r_{i,j}}+2} \) and \( E_{i,1}^{j} \subset B_{i,1}^{j} \), we have $B_{i,1}^{j} = B_{i,1}^{j+1} \cup E_{i,1}^{j}.$ This implies that $\bigcup_{j=1}^{t_i}E_{i,1}^{j}=B_{i,1}$.

Thus, \( \bigcup_{j=1}^{t_i} S^{min}(E_{i,1}^{j}) \) is a selective subset of \( B_{i,1} \) for \( 1\leq i\leq k \). Consequently, the set 
\[
\bigcup_{i=1}^{k}\bigcup_{j=1}^{t_i} S^{min}(E_{i,1}^{j})
\]
forms a selective subset of \( G \).\qed
\end{proof}
\section{Proof of the Lemma~\ref{lemma10}}
%
%
\begin{proof}
According to Lemma~\ref{lemma2}, for each vertex \( v \in E_{i,1}^{j} \subseteq F_{i,j} \), either \( v \in S^{min}(E_{i,1}^{j}) \) or at least one of its neighbors in \( F_{i,j} \) must be in \( S^{min}(E_{i,1}^{j}) \). However, vertices in \( F_{i,j}\setminus S^{min}(E_{i,1}^{j}) \) may not have any adjacent vertex in \( S^{min}(E_{i,1}^{j}) \), as shown in the proof of Theorem~\ref{th3}. Therefore, \( S^{min}(E_{i,1}^{j}) \) may not itself be a dominating set of \( F_{i,j} \), but its size is at most that of a minimum dominating set of \( F_{i,j} \).

Moreover, by the proof of Theorem~\ref{th3}, if \(u \in F_{i,j}\setminus S^{min}(E_{i,1}^{j})\) and $C(u)=\ell$, then \(u\) has a nearest neighbor of the same color in \(S^{min}(E_{i,1}^{j}) \cup (V(G) \setminus V_{\ell})\). That is, every vertex in \(F_{i,j}\setminus S^{min}(E_{i,1}^{j})\) satisfies the condition of a selective subset.

Here, \(\gamma(F_{i,j})\) denotes the size of a minimum dominating set of \(F_{i,j}\), 
and \(\alpha(F_{i,j})\) denotes the size of a maximum independent set of \(F_{i,j}\).

Thus, we have
\[
|S^{min}(E_{i,1}^{j})| \leq \gamma(F_{i,j}).
\]
Since every maximum independent set in \( F_{i,j} \) is a dominating set, it follows that
\[
\gamma(F_{i,j}) \leq \alpha(F_{i,j}).
\]
Consequently,
\[
|S^{min}(E_{i,1}^{j})| \leq \alpha(F_{i,j}),
\]
as claimed.
\end{proof}

\section{Proof of the Lemma~\ref{lemma11}}
\begin{proof}
By the definition of unit disk graphs, we have:
\begin{equation}
  u,v \in \n_{i,j} ^{r}[v_i^{j}, B_{i,3}^{j}] \quad \iff \quad \| f(u) - f(v) \| \leq 2r.
\end{equation}
Thus, \( I^{r} \) consists of mutually disjoint unit disks within a disk of radius \( 2r+1 \) centered at \( f(v) \). 

Since the maximum number of such disjoint unit disks that can fit in a disk of radius \( 2r+1 \) is given by:
\[
\lvert I^{r} \rvert \leq \frac{\pi (2r+1)^2}{\pi} = (2r+1)^2 = \mathcal{O}(r^2),
\]
the result follows.\qed
\end{proof}
\section{Proof of the Lemma~\ref{lemma20}}
\begin{proof}
For any \( r < \overline{r_{i,j}} \), condition (\textcolor{blue}{\ref{eq1}}) implies the following inequalities depending on whether \( r \) is even or odd:

When \( r \) is even:
\begin{equation}
\begin{split}
(2r+1)^2 \geq \lvert S^{min}(X_{i,j}^{r+2})\rvert > \delta \lvert S^{min}(X_{i,j}^{r})\rvert  \\
> \dots > \delta ^{\frac{r}{2}} \lvert S^{min}(X_{i,j}^{0})\rvert = \delta^{\frac{r}{2}}.
\end{split}
\end{equation}

When \( r \) is odd:
\begin{equation}
\begin{aligned}
(2r+1)^2 &\geq \lvert S^{min}(X_{i,j}^{r+2})\rvert > \delta \lvert S^{min}(X_{i,j}^{r})\rvert  \\
&> \dots > \delta ^{\frac{r}{2}} \lvert S^{min}(X_{i,j}^{1})\rvert = \delta^{\frac{r}{2}}.
\end{aligned}
\end{equation}
Since \( \delta > 1 \), the above inequalities hold until condition (\textcolor{blue}{\ref{eq1}}) becomes violated. Therefore, when \( r < \overline{r_{i,j}} \) the above inequalities are violated (that is, $(2\overline{r_{i,j}}+1)^{2} <  \delta^{\frac{\overline{r_{i,j}}}{2}}$), and this first violation determines the value of \(\overline{r_{i,j}} \), which depends only on \( \delta \), not on the size of $V(U)$. To prove this, we now show that $(2d+1)^2 < (1+\epsilon)^d$ where $d=\mathcal{O}\Big(\frac{1}{\epsilon^2} \log \frac{1}{\epsilon}\Big)$. Note that we use $d$ instead of $\overline{r_{i,j}}$ for simplicity of calculation.

For any \( d \geq 1 \), we first prove that
\[
(2d+1)^2 < (3d)^2 < (1+\epsilon)^d.
\]
The first inequality holds for all \( d \geq 1 \). To prove the second inequality, it suffices to show:
\[
2\log (3d) < d\log (1+\epsilon).
\]
This implies:
\[
\frac{2}{d} \log (3d) + \epsilon^2 < \log (1+\epsilon) + \epsilon^2.
\]

Substituting \( d = \frac{1}{\epsilon^2} \log \frac{1}{\epsilon} \) into $\frac{2}{d} \log (3d) + \epsilon^2$, we obtain:
\begin{equation}\label{eq6}
\begin{aligned}
\frac{2}{d}(\log 3+ \log d)+\epsilon^2 &= \frac{2\epsilon^2}{\log \frac{1}{\epsilon}}
\left(\log 3+\log \frac{1}{\epsilon^2} + \log\log\frac{1}{\epsilon}\right) + \epsilon^2  \\
&= 2\epsilon^2 \left( \frac{\log 3}{\log \frac{1}{\epsilon}}
+ 2\frac{\log \frac{1}{\epsilon}}{\log \frac{1}{\epsilon}}
+ \frac{\log \log \frac{1}{\epsilon}}{\log \frac{1}{\epsilon}}
+ \frac{1}{2} \right) \\
&< 2\epsilon^2 (\log 3+3.5).
\end{aligned}
\end{equation}
Since \( 0 < \epsilon < \frac{1}{10} \), it follows that:
\begin{equation}\label{eq7}
2\epsilon^2 (\log 3 + 3.5) < 10\epsilon^2 < \epsilon,
\end{equation}
As established in~\textcolor{blue}{\cite{GrahamKnuthPatashnik1998}}, for \( 0 \leq \epsilon < \frac{1}{2} \), the following inequality holds
\begin{equation}\label{eq8}
    \epsilon \leq \log (1+\epsilon) + \epsilon^2
\end{equation}
Substituting inequalities (\textcolor{blue}{\ref{eq7}}) and (\textcolor{blue}{\ref{eq8}}) into inequality (\textcolor{blue}{\ref{eq6}}), we obtain
\[
\frac{2}{d} \log (3d) + \epsilon^2 < \log (1+\epsilon) + \epsilon^2.
\]
Moreover, by Theorem~\textcolor{blue}{\ref{th6}}, \( S^{min}(E_{i,1}^{j}) \) can be computed in \( \mathcal{O}(n^{d^2}) \), where \( d(\epsilon) = \mathcal{O}(\frac{1}{\epsilon^2} \log \frac{1}{\epsilon}) \) and \( 0 < \epsilon < \frac{1}{10} \).\qed
\end{proof}

\section{Proof of the Lemma~\ref{lemmaapx}}
\begin{proof}
If $D$ is a dominating set of $G$ of size $t$, then $S = D \cup V(X)$ forms a selective subset of $H$. This holds because each chord in $V(X)$ forms a block, and by Lemma~\textcolor{blue}{\ref{lemma2}}, we have $V(X) \subset S$. Moreover, for every vertex $v \in V(H) \setminus S$, we have $\dist (v, S) = \dist (v, V(X)) = 1$.

Conversely, if $S$ is a selective subset of size $m+56n+4+t$, then by Lemma~\textcolor{blue}{\ref{bubai1}}, we have $V(X) \subset S$, which implies that $\lvert S \setminus V(X) \rvert = t$. The set $S \setminus V(X)$ forms a dominating set of $G$, since for each vertex $v \in  V(G)$, we have $\dist (v, V(X)) = 1$ in $H$. Therefore, either $v\in S \setminus V(X)$ or at least one neighbor of $v$ from $V(G)$ must be included in $S \setminus V(X)$. This completes the proof.\qed
\end{proof}
\section{Proof of the Theorem~\ref{bubailast}}
\begin{proof}
Using Lemma~\textcolor{blue}{\ref{lemmaapx}}, we have \( S^{min}(H) = m + 56n + 4 + \gamma(G) \). 
We use this equation in Theorem~\textcolor{blue}{\ref{theoremcircle}} for any $0 < \alpha < 1$, and we get 
\[
\text{SAT}(\phi) = 1 \implies \lvert S^{min}(H)\rvert \leq 16n + 2 + (m + 56n + 4),
\]
\[
\text{SAT}(\phi) < \alpha \implies \lvert S^{min}(H)\rvert > 16n + 2 + (m + 56n + 4) + \frac{(1 - \alpha)m}{8}.
\]
This clearly indicates that the gap-preserving reduction works. Therefore, this reduction shows that \mss is \apxh in circle graphs. Hence, the selective subset problem does not admit a \ptas\ unless \( \npp = \np \).\qed
\end{proof}
\end{document}